%% file: new-main.tex
\documentclass{article}
\usepackage{graphicx} 

\usepackage{amsmath,amsthm,amssymb,amsfonts}
\usepackage{caption}
\usepackage{subcaption}
\usepackage{bbm}
\usepackage{graphicx}
\usepackage{multirow}
\usepackage{xspace}
\usepackage[colorlinks=true,citecolor=blue,linkcolor=blue]{hyperref}
\usepackage{cleveref}
\usepackage{siunitx,booktabs}
\usepackage{todonotes}

\usepackage[letterpaper,top=2cm,bottom=2cm,left=3cm,right=3cm,marginparwidth=1.75cm]{geometry}

\usepackage{tikz}
\usetikzlibrary{arrows,shapes,automata, calc, chains, matrix, trees, positioning, scopes}

\newtheorem{theorem}{Theorem}
\newtheorem{proposition}[theorem]{Proposition}
\newtheorem{corollary}[theorem]{Corollary}
\newtheorem{lemma}[theorem]{Lemma}

\newtheorem*{problem*}{Problem}

\numberwithin{theorem}{section}

\theoremstyle{definition}
\newtheorem{example}[theorem]{Example}

\newcommand{\comment}[1]{}

\newcommand{\Until}{\mathrel{{\mathrm{U}}}}
\newcommand{\Next}{\mathrm{X}}
\newcommand{\Event}{\mathrm{F}}
\newcommand{\Glob}{\mathrm{G}}
\newcommand{\LLL}{\mathrm{L}}

\newcommand{\Aa}{\mathcal{A}}
\newcommand{\BB}{\mathcal{B}}
\newcommand{\CC}{\mathcal{C}}
\newcommand{\DD}{\mathcal{D}}

\newcommand{\PP}{\mathcal{P}}

\newcommand{\NP}{\textnormal{\textsf{NP}}}
\newcommand{\PSPACE}{\textnormal{\textsf{PSPACE}}}
\newcommand{\NPSPACE}{\textnormal{\textsf{NPSPACE}}}

\newcommand{\LTL}{LTL$_f$}

\DeclareMathOperator{\paths}{paths}
\DeclareMathOperator{\sets}{sets}

\newcommand{\ore}[0]
{\mathrm{R^{\ell < \ell}_{\sets}}}
\newcommand{\orz}[0]
{\mathrm{R^{\ell \le \ell}_{\sets}}}
\newcommand{\orep}[0]
{\mathrm{R^{\ell < \ell}_{\paths}}}
\newcommand{\orzp}[0]
{\mathrm{R^{\ell \le \ell}_{\paths}}}
\newcommand{\ordp}[0]
{\mathrm{R^{\ell < f}_{\paths}}}

\newcommand{\CS}[1]
{\textsc{Constrained-$#1$-Sync}}
\newcommand{\CW}[1]
{\textsc{Constrained-$#1$-Word}}

\DeclareMathOperator{\first}{first}
\DeclareMathOperator{\last}{last}

\usepackage{pifont}
\usepackage{soul}
\setstcolor{red}

\title{Traversing automata with current state uncertainty\\ under \LTL{} constraints}
\author{Andrew Ryzhikov$^1$, Petra Wolf$^2$}
\date{$^1$Department of Computer Science, University of Oxford, UK\\
$^2$ LaBRI, Université de Bordeaux, France}

\begin{document}

\maketitle

\begin{abstract}
In this paper, we consider a problem which we call \LTL{} model checking on paths:
given a DFA~$\Aa$ and a formula $\phi$ in LTL on finite traces, does there exist a word $w$ such that every path starting in a state of $\Aa$ and labeled by $w$ satisfies $\phi$? The original motivation for this problem comes from the constrained parts orienting problem, introduced in [Petra Wolf, ``Synchronization Under Dynamic Constraints'', FSTTCS 2020], where the input constraints restrict the order in which certain states are visited for the first or the last time while reading a word $w$ which is also required to synchronize~$\Aa$. We identify very general conditions under which \LTL{} model checking on paths is solvable in polynomial space. For the particular constraints in the parts orienting problem, we consider \PSPACE{}-complete cases and one \NP{}-complete case. The former provide very strong lower bound for \LTL{} model checking on paths. The latter is related to (classical) \LTL{} model checking for formulas with the until modality only and with no nesting of operators. We also consider \LTL{} model checking of the power-set automaton of a given DFA, and get similar results for this setting. For all our problems, we consider the case where the required word must also be synchronizing, and prove that if the problem does not become trivial, then this additional constraint does not change the complexity. 
\end{abstract}

\section{Introduction}
\input{sec-introduction}

\section{Main definitions}\label{sec:definitions}
\input{sec-definitions}

\section{Polynomial space upper bounds}\label{sec:in-pspace}
\input{sec-pspace}

\section{\PSPACE{}-complete constraints}\label{sec:pspace-c}
\input{sec-subsets}

\section{\NP{}-complete constraints}\label{sec:np-c}
\input{sec-graph-traversal}

\section{Conclusions and open problems}\label{sec:conclusions}
\input{sec-conclusions}

\subsection*{Acknowledgements}
We thank Antonio Di Stasio and Przemys{\l}aw Wa{\l}\c{e}ga for useful remarks and fruitful conversations. Andrew Ryzhikov is supported by the European Research Council (ERC) under the European Union’s Horizon 2020 research and innovation programme (Grant agreement No. 852769, ARiAT). Petra Wolf is supported by the French ANR, project ANR-22-CE48-0001 (TEMPOGRAL).

\bibliographystyle{alpha}
\bibliography{sync_under_order}

\end{document}

%% file: sec-introduction.tex
Synchronization is a classical and well-studied way of regaining control over an automaton if its current state is unknown. Given a complete deterministic finite automaton (complete DFA), we say that a finite word is synchronizing for it, if it maps all its states to one particular state. Such a word can be seen as a word resolving the current state uncertainty: if the structure of the DFA is known, then after reading a synchronizing word the current state of the DFA is also known. If a DFA admits a synchronizing word, it is also called synchronizing.

\paragraph*{Orienting parts.} One of the oldest applications of synchronizing automata is the problem of designing parts orienters, which are simple robots or machines that get a number of identical objects in various orientations and need to bring them all in the same orientation, usually using conveyor belts with obstacles. This approach is relatively cheap, since it does not require introducing any sensors that register the orientation of each object, and instead just makes sure that the final orientation is the same regardless of the initial orientation of the object, see~\cite{DBLP:journals/tcs/AnanichevV04} for an example. In his pioneering work, Natarajan \cite{DBLP:conf/focs/Natarajan86} modeled parts orienters as complete DFAs where states correspond to possible orientations of a part and letters correspond to applying different modifiers or obstacles. Because of their shape and design, those modifiers can have different effect on the parts depending on the orientation of the parts. In this context, he studied automata which were later called \emph{orientable}~\cite{DBLP:conf/lata/Volkov08}.
Many different classes of automata have since been studied regarding their synchronization behavior. We refer to~\cite{DBLP:conf/lata/Volkov08,beal_perrin_2016,JALC20,Kari2021} for an overview.

\begin{figure}
    \centering
    \includegraphics[width=0.4\textwidth]{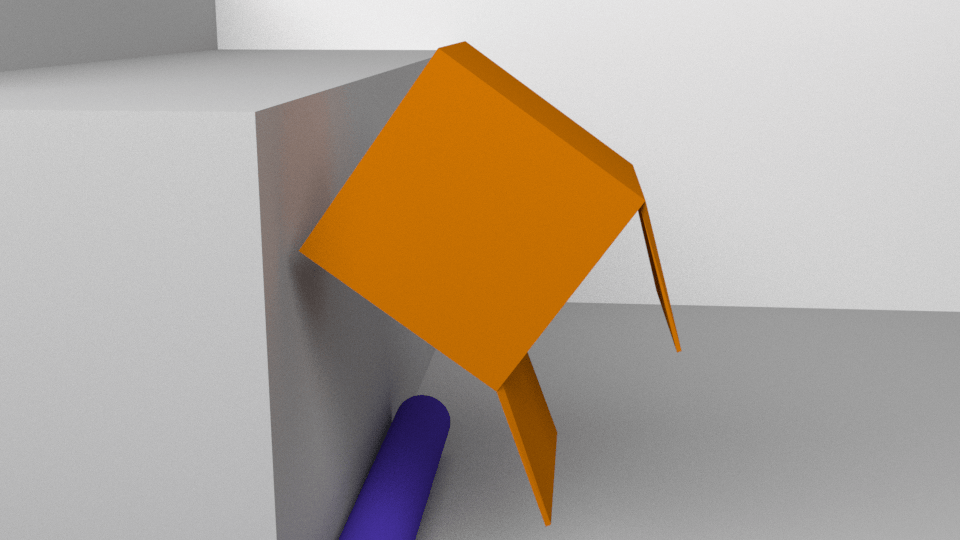}
    \includegraphics[width=0.4\textwidth]{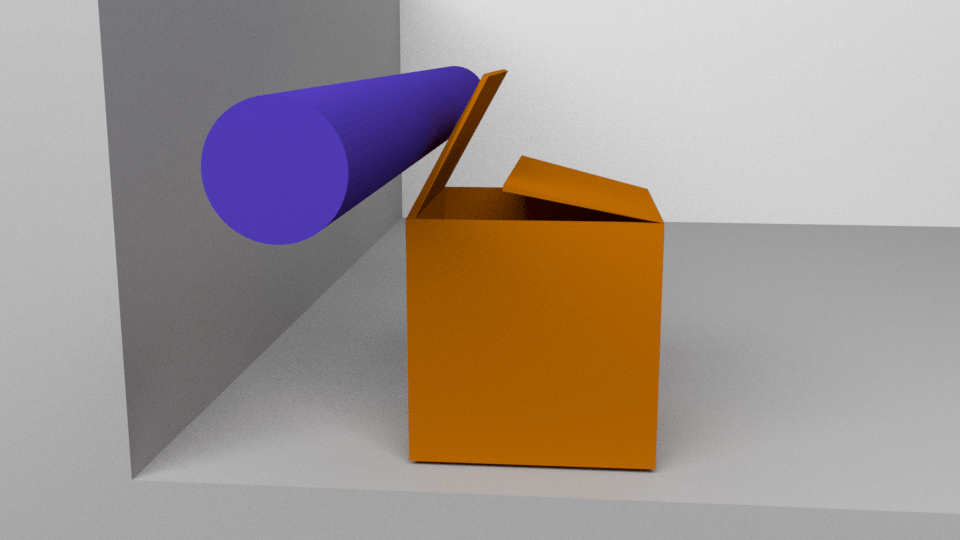}
    \caption{Illustration of how different modifiers of a parts orienter can have an impact on the part itself. Here, the part consists of a box with lids that open when the box is rotated and need to be closed again by a different modifier.}
    \label{fig:box}
\end{figure}
The original motivation of designing a parts orienter was revisited in~\cite{DBLP:journals/ijfcs/TurkerY15} where T{\"{u}}rker and Yenig{\"{u}}n modeled the design of an assembly line, which again brings a part from an unknown orientation into a known orientation, with different costs of modifiers. 
What has not been considered so far is that different modifiers can have different impacts on a part, depending on the the history of its orientations. For example, if the part is a box with a fold-out lid, turning it upside-down will cause the lid to open as depicted in \Cref{fig:box}. In order to close the lid, one might need another modifier such as a low bar which brushes the lid and closes it again. To specify that a parts orienter must deliver the box facing upward with a closed lid, one needs to encode the property ``when the box is in the state \emph{facing down}, it must later be in the state \emph{lid closed}''. But this does not stop us from opening the lid again in the future, so we need to be more precise and instead encode the property ``after \textbf{the last time} the box was in the state \emph{facing down}, it must visit the state \emph{lid closed} at least once''. 
In order to model these types of constraints in the automaton abstraction of a parts orienter, we associate with each word several different traversal relations describing how it traverses the states of a DFA. We then implement the conditions mentioned above by requiring the traveral relation associated with the word to agree with a given binary relation $R$ on the states of the DFA. We consider two different cases: restricting the traversal relation $\orep$, intuitively, requires visiting the state \emph{lid closed} no matter what, restricting the traversal relation $\orzp$ only requires closing the lid if it was opened before that.

To motivate the third traversal relation we consider, $\ordp$, let us again picture the box with a lid, but this time the box initially contains some water. We would like to have the box in a specific orientation with the lid open but the water must not be shed while orientating. We have a modifier that opens the lid and a modifier which rotates the box. Clearly, we do not want the box to face downwards after the lid has been opened. So, we encode the property ``once the state \emph{lid open} is reached, the state \emph{facing downwards} must never be visited again''. 

\paragraph{Controlling multiple automata sharing a common resource.} Consider now the scenario where several copies of the same DFA $\Aa$ are controlled by the same input, and we know precisely the set~$S$ of current states of these copies. If $S$ is synchronizing, that is, if there exists a word bringing all the states of $S$ to one particular state, then after applying a synchronizing word for it we know that the current state of all the copies is the same. However, the DFAs can have some side effects on the resources which are shared between all of them. As an example, let $p$ be a state of~$\Aa$ such that visiting it requests a shared resource and marks it as busy for all DFAs, and $q$ be a state such that visiting it releases the resource and marks it as free for all DFAs. It is natural to ask that after the last moment of time when the resource was requested by at least one copy of $\Aa$, it was released by a copy, not necessarily the same one. In other words, it must not happen that some copy required the resource and no copy released it afterwards. This restriction can be modeled by considering sets of possible current states of~$\Aa$, and is formalized by including the pair $(p, q)$ in the binary relation restricting the traversal relation $\ore$ defined in \Cref{sec:definitions}. We also consider a less strict traversal relation $\orz$ where the last visits of $p$ and $q$ are allowed to happen at the same time. We show that checking whether a word satisfying such requirements exists is \PSPACE{}-complete, even if~$S$ is the set of all states of $\Aa$.

Another possible interpretation of this setting is that we are given a DFA $\Aa'$ and a constraint, and want to check if there exists a word accepted by $\Aa'$ such that the path labeled by it in $\Aa'$ satisfies this constraint. If $\Aa'$ is provided as a part of the input, this problem is solvable in polynomial time for any fixed order constraint. However, we show that if $\Aa'$ is defined as the power-set automaton of a DFA $\Aa$ provided in the input, the problem becomes \PSPACE{}-complete, even for a fixed constraint. Hence, with such a succinct representation the complexity of the problem increases drastically, but it is still better than the trivial approach of explicitly constructing the power-set automaton, which takes exponential time and space.

The traversal relations $\orep$, $\orzp$ and $\ordp$ mentioned above also make sense in the setting where we have a DFA whose structure is known to us, but its current state is not. Constraining them ensures that whichever state we start with, when applying the sequence of input commands the requirements are satisfied. For example, if there is only one copy of a DFA releasing and requesting a resource, and its current state is unknown, we can ask for a word which does not request the resource after it is released no matter in which state the DFA was, and which brings it to a known state.

For $\ore$, $\orep$ and $\orzp$ it is also natural to ask for a word which satisfies the conditions but is not necessarily synchronizing. In this case, we only care that the resource is eventually released, but do not require that all the copies are in the same state afterwards. We prove that dropping the requirement that the word must be synchronizing does not change the complexity of the problem. However, satisfying the constraints imposed on $\orz$ and $\ordp$ with a word which is not required to be synchronizing is trivial: they are satisfied by the empty word. 

\paragraph*{LTL on finite traces.} Linear temporal logic (LTL) is a popular and expressive formalism for specifying properties of systems \cite{Baier2008}. Classical LTL is defined on infinite traces and is more suitable for describing the behavior of reactive systems. Its variation, LTL on finite traces \cite{DeGiacomo2013}, denoted \LTL{}, is more suitable for describing the behavior of systems performing  tasks which eventually terminate. The difference between LTL and \LTL{} is discussed in \cite{DeGiacomo2014}. All the constraints on the traversal relations that we consider, as well as their Boolean combinations, can be expressed as formulas of \LTL{} in a suitable setting. Moreover, the constructions showing computational complexity lower bounds for these constraints provide very strong lower bounds on the complexity of model checking in our setting.

The classical \LTL{} model checking problem is, given an \LTL{} formula $\phi$ and a transition system with a specified initial state $s$ and without any labels of the transitions, check if there exists a path in the transition system starting in $s$ and satisfying $\phi$. Our setting differs in that there is no specified initial state, and the transitions are labeled by letters of the alphabet. In more detail, we consider either a set of paths starting in each state of the DFA and labeled by the same word (\LTL{} model checking on paths), or a unique path in the power-set automaton labeled by this word (\LTL{} model checking on sets). However, we show that the polynomial space algorithm for the classical setting can be adapted to our settings without an increase in its computational complexity. In particular, this provides polynomial space upper bounds for the problem of satisfying all the restrictions discussed above. On the other hand, in the classical setting model checking becomes solvable in polynomial time if the \LTL{} formula is fixed. In contrast, we show that there exist fixed formulas such that \LTL{} model checking on paths and on sets, as well as their synchronized versions, remain \PSPACE{}-complete, even for DFAs over constant-size alphabet. For synchronized \LTL{} model checking on sets, it follows from the \PSPACE{}-completeness of synchronization with constraints on $\orz$ for a constraint relation of constant size, originally proved in \cite{Wolf2020}. For \LTL{} model checking on paths it was explicitly stated and proved in \cite{Bertrand2023} when the alphabet size grows with the number of states. See the subsection ``Existing results and our contributions'' below for a more detailed discussion of the differences between the settings of \cite{Bertrand2023} and this paper.

\paragraph*{Related work.}
The problem of checking whether for a given DFA $\Aa=(Q, \Sigma, \delta)$ there exists a synchronizing word can be solved in time $\mathcal{O}(|Q|^2 \cdot |\Sigma|)$ ~\cite{DBLP:journals/siamcomp/Eppstein90,DBLP:conf/lata/Volkov08}. The problem of finding such a word is solvable in time~$\mathcal{O}(|Q|^3 + |Q|^2 \cdot |\Sigma|)$. In comparison, if we only ask for a subset of states $S \subseteq Q$ to be synchronized, the problem becomes \PSPACE-complete~\cite{DBLP:conf/dagstuhl/Sandberg04}. These two problems have been investigated for several restricted classes of automata involving orders on states. 
Here, we want to mention the class of \emph{oriented} automata, whose states can be arranged in a cyclic order which is preserved by all transitions. This model has been studied among others in \cite{DBLP:conf/focs/Natarajan86,DBLP:journals/siamcomp/Eppstein90,DBLP:journals/tcs/AnanichevV04,DBLP:journals/fuin/RyzhikovS18,DBLP:conf/lata/Volkov08}. If the order on the states is linear instead of cyclic, we get the class of monotonic automata which was studied in \cite{DBLP:journals/tcs/AnanichevV04,DBLP:journals/fuin/RyzhikovS18}.

A (complete or partial) DFA $\Aa = (Q, \Sigma, \delta)$ is called \emph{partially ordered}~\cite{Brzozowski1980} or \emph{weakly acyclic}~\cite{DBLP:journals/tcs/Ryzhikov19a} if there exists an ordering of the states $q_1, q_2, \dots, q_n$ such that if $\delta(q_i, a) = q_j$ for some letter $a \in \Sigma$, then $i \le j$.
In other words, all cycles in a partially ordered DFAs are self-loops.
Each synchronizing partially ordered complete DFA admits a synchronizing word of linear length~\cite{DBLP:journals/tcs/Ryzhikov19a}.
Going from complete DFAs to partial DFAs brings a jump in complexity. For example, the so called \emph{careful synchronization} problem for partial DFAs asks for a word that is defined on all states and that brings all states to one particular state.
This problem is \PSPACE-complete for partial DFAs, even over binary alphabets and with only one undefined transition~\cite{DBLP:journals/mst/Martyugin14, DBLP:conf/wia/Martyugin12}. 

Linear temporal logic (LTL) was introduced as a formalism for verifying properties of programs \cite{Pnueli1977, Pnueli1981}. The LTL model checking problem is, given a transition system and an LTL formula, check if there exists an infinite path in the transition system satisfying the formula. This problem is \PSPACE{}-complete and becomes solvable in polynomial time if the formula is fixed \cite{Baier2008}. In \cite{Markey2004}, the computational complexity of many fragments of LTL was investigated. LTL on finite traces, denoted \LTL{}, was introduced in \cite{DeGiacomo2013} to describe properties of finite paths. The complexity of satisfiability of formulas from different fragments in \LTL{} is studied in \cite{Fionda2018}, and a framework for \LTL{} satisfiability checking is proposed in \cite{Li2014}. 

\paragraph*{Existing results and our contributions.} The parts orienting problem and traversal constraints were introduced in \cite{Wolf2020}, a preliminary conference version of this paper. See \Cref{subsec-order-defs} for the discussion on the differences of the notation between \cite{Wolf2020} and the current paper.  In \cite{Wolf2020}, it was shown that checking the existence of a synchronizing word with constraints on $\ore$, $\orz$ and $\ordp$ is \PSPACE{}-complete over constant-size alphabet, and for constraints on $\ore$ and $\orz$ it stays \PSPACE{}-complete even for constraint relations of constant size. It was also shown that checking the existence of a synchronizing word with constraints on $\orep$ is in \NP{}, and with constraints on $\orzp$ is \NP{}-hard. 

This setting was later extended to arbitrary \LTL{} formulas on paths in \cite{Bertrand2023}. There, is was proved that checking the existence of a synchronizing word satisfying a given \LTL{} formula is solvable in polynomial space, and there exist fixed formulas for which it stays \PSPACE{}-complete. It was also proved that synchronization with constraints on $\orzp$ is \PSPACE{}-complete. However, all the mentioned complexity lower bounds in \cite{Bertrand2023} require alphabets of variable size, and atomic propositions that are more general than just the names of the states.

In this paper, we generalize, unify and improve the results of \cite{Wolf2020} and \cite{Bertrand2023}. In \Cref{sec:in-pspace}, we recall the \PSPACE{} upper bound for checking the existence of a synchronizing word under \LTL{} constraints in \cite{Bertrand2023}, and discuss more general conditions where it can be applied. The most interesting cases are what we call \LTL{} model checking on paths and sets, as well as their synchronized versions. As an application, we show that this generalized algorithm solves synchronization with constraints on $\ore$ and $\orz$ in polynomial space.

In \Cref{sec:pspace-c}, we prove that for \LTL{} model checking on paths and sets, as well as their synchronized versions, there exist fixed \LTL{} formulas for which they remain \PSPACE{}-hard, even for DFAs over constant-size 
 alphabets (\Cref{thm:mc-sets-main} and \Cref{thm:mc-paths-main}). Moreover, in all our complexity lower bounds in this paper, we assume that the only atomic proposition which is true in a state is the name of this state. Hence, the properties such as ``the state $p$ belongs to a set $S$'' cannot be expressed by fixed \LTL{} formulas. The lower bounds in \Cref{sec:pspace-c} come from \PSPACE{}-hardness of synchronization with constraints on $\ore$, $\orz$, $\orep$ and $\ordp$.

In \Cref{sec:np-c}, we close an open problem stated in \cite{Wolf2020} and show that synchronization with constraints on $\orep$ is \NP{}-complete (\Cref{thm:l-l-npc}). We show a connection of this problem to (classical) \LTL{} model checking where the formula contains only the until modality without nesting, which we prove to be \NP{}-hard (\Cref{thm:until-npc}). The expressivity of \LTL{} formulas with limited nesting of until operators is studied in \cite{Therien2004}.

%% file: sec-definitions.tex
A \emph{complete deterministic finite semi-automaton} $\Aa = (Q, \Sigma, \delta)$, which we simply call a DFA in this paper, consists of the finite set $Q$ of states, the finite alphabet $\Sigma$, and the transition function $\delta: Q \times \Sigma \to Q$.  The transition function $\delta$ is generalized to finite words in the usual way. It is further generalized to sets of states $S \subseteq Q$ as $\delta(S, w) = \{\delta(q, w)\mid q \in S\}$. We sometimes refer to $\delta(S, w)$ as~$S.w$. We often talk about active states. Given a set $S$ of previously active states, we say that $q$ is active after the application of $w$ if $q \in S.w$. 
We denote by $|S|$ the size of the set~$S$. We might identify singleton sets with their elements. If for some $w\in \Sigma^*$ and a set of states $S\subseteq Q$, $|S.w| = 1$, we say that the word $w$ \emph{synchronizes} the set~$S$.
Given a word $w$, we denote by $|w|$ the length of $w$, by $w[i]$ the~$i^\text{th}$ letter of $w$ (or the empty word $\epsilon$ if $i = 0$) and by $w[i..j]$ the factor of $w$ from letter $i$ to letter~$j$. 
For a state~$q$, we call the sequence of active states $q.w[1..i]$ for $0 \leq i \leq |w|$ the \emph{path} induced by $w$ starting at~$q$.
A \emph{DFA acceptor} $\Aa = (Q, \Sigma, \delta, i, F)$ is a DFA with a chosen initial state $i \in Q$ and a set of final state $F \subseteq Q$. The \emph{language accepted by} $\Aa$ is the set of words $w$ such that $i.w \in F$.

Let $w \in \Sigma^*$ and $q \in Q$. We define $\paths_\Aa(q, w)$ to be the path in $\Aa$ starting in $q$ and labeled by $w$. Given $S \subseteq Q$, define $\paths_\Aa(S, w) = \cup_{q \in S} \{\paths_\Aa(q, w)\}$. 
Define the \emph{power-set automaton} $\PP(\Aa) = (2^Q, \Sigma, \delta')$ of $\Aa$ as usually: the set of states of $\PP(\Aa)$ is the set $2^Q$ of all subsets of $Q$, and the transition function is defined as $\delta'(S, a) = \{\delta(q, a) \mid q \in S\}$. We define $\sets_\Aa(S, w)$ to be the single path in~$\PP(\Aa)$ starting in the state $S$ and labeled by $w$.

A \emph{partial DFA} is a DFA $\Aa = (Q, \Sigma, \delta)$ where the transition function $\delta$ is allowed to be partial, that is, undefined for some input values. Its extensions to $Q \times \Sigma^*$ and further to $2^Q \times \Sigma^*$ are defined in the same way as for complete DFAs. We say that a word $w$ \emph{carefully synchronizes} $S \subseteq Q$ if the action of $w$ is defined for every state in $S$ and $|\delta(S, w)| = 1$. If $w$ carefully synchronizes $Q$, we say that it \emph{carefully synhronizes} $\Aa$, and then call both $w$ and $\Aa$ \emph{carefully synchronizing}.

The \emph{underlying digraph} of a DFA $\Aa$ is obtained by forgetting the labels of transitions in $\Aa$. Given a digraph $G = (V, E)$, we call a \emph{path} a sequence of vertices interleaved with edges
$$\rho = v_1 \xrightarrow[]{e_1} v_2 \xrightarrow[]{e_2} \dots \xrightarrow[]{e_n} v_{n+1}$$
such that $e_i = (v_i, v_{i + 1})$ for $1 \le i \le n$. We emphasize that a path can have repeating vertices. We say that a vertex $v$ is reachable from a vertex $u$ if there exists a path from $u$ to $v$. The relation ``$u$ and $v$ are reachable from each other'' is an equivalence relation, and a class of this relation is called a \emph{maximal strongly connected component}. A maximal strongly connected component is called a \emph{sink} if only vertices belonging to it are reachable from its vertices. A digraph is called \emph{strongly connected} if every its vertex is reachable from every other vertex. A DFA is called \emph{strongly connected} if its underlying digraph is strongly connected.

We expect the reader to be familiar with basic concepts of automata theory, computational
complexity and model checking, and refer to the textbooks~\cite{DBLP:books/daglib/0086373} and \cite{Baier2008}
 as a reference.

\subsection{\LTL{} model checking}\label{subsec:\LTL{}-defs}

Let $P$ be a finite set, which we call the set of \emph{atomic propositions}. The \emph{linear temporal logic on finite traces (\LTL{})} over $P$ is the set of formulas defined inductively as follows. The symbol $\top$ representing  the constant true, and every element of $P$ are in \LTL{}. Moreover, if $\phi, \psi$ are in \LTL{}, then $\neg \phi, \phi \lor \psi, \phi \land \psi, \Next \phi, \phi \Until \psi, \Event \phi, \Glob \phi$, as well as all Boolean combinations of this formulas are in \LTL{}. The operators $\Until, \Next, \Event, \Glob$ stand for until, next, finally and globally, and their semantics is the standard semantics on finite traces \cite{DeGiacomo2013,Bertrand2023}, which are in our setting paths in a DFA. The size $|\phi|$ of the formula $\phi$ is the length of a bit string encoding it.

We also consider restricted formulas of \LTL{}. Namely, given a modality $T$, we denote by  $\LLL^+(T)$ the set of \LTL{} formulas where negation is allowed only in front of atomic propositions, and the only allowed modality is $T$.

We defined the following problems.

\begin{problem*}
  \textsc{\LTL{} model checking on paths} \\
  \textbf{Input:} A DFA $\Aa = (Q, \Sigma, \delta)$ and an \LTL{} formula $\phi$ over the set $P = Q$. \\
  \textbf{Output:} Yes, if and only if there exists a word $w$ such that for every $\pi \in \paths(Q, w)$ we have $\pi \models \phi$.
\end{problem*}

\begin{problem*}
  \textsc{\LTL{} model checking on sets} \\
  \textbf{Input:} A DFA $\Aa = (Q, \Sigma, \delta)$ and an \LTL{} formula $\phi$ over the set $P = Q$. \\
  \textbf{Output:} Yes, if and only if there exists a word $w$ such that for the path $\pi = \sets(Q, w)$ we have $\pi \models \phi$.
\end{problem*}

The problems \textsc{Synchronized \LTL{} model checking on paths} (\textsc{on sets}) are defined similarly with an additional requirement that $w$ must synchronize $\Aa$. As discussed in \Cref{subsec-order-defs}, it is not necessarily the case that an \LTL{} fomula which is true on paths is also true on sets, or vice versa.

\subsection{Traversal relations}\label{subsec-order-defs}

We now describe different traversal relations $\mathrm{R}(w)$ which are induced by the traversal of a word $w$ through the states of a DFA. The first two relations are defined by the last visits of the states to each other, while the third relation is defined by the connection between first and last visits. Given a DFA $\Aa = (Q, \Sigma, \delta)$ and an arbitrary relation $R \subseteq Q^2$ which we call a \emph{constraint relation}, we say that a traversal relation $\mathrm{R}(w)$ \emph{agrees} with $R$ if and only if $R \subseteq \mathrm{R}(w)$. We then consider the following problem for each traversal relations $\mathrm{R}(w)$ introduced below.

\begin{problem*}\label{prob-sync-under}
  \CS{\mathrm{R}(w)} \\
  \textbf{Input:} A DFA $\Aa = (Q, \Sigma, \delta)$ and a constraint relation $R \subseteq Q^2$. \\
  \textbf{Output:} Yes, if and only if there exists a word $w \in \Sigma^*$ such that $R \subseteq \mathrm{R}(w)$ and $w$ synchronizes $\Aa$.
\end{problem*}

We also consider version of this problem where $w$ is not required to synchronize $\Aa$. We call this variant \CW{\mathrm{R}(w)}. In \cite{Wolf2020} and \cite{Bertrand2023}, the problem \CS{\mathrm{R}(w)} was called \textsc{Sync-Under-$\lessdot_w$}.

\comment{
\begin{problem*}\label{prob-subset-sync-under}
  \textsc{Subset-Sync-Under-$\lessdot_w$} \\
  \textbf{Input:} A DFA $\Aa = (Q, \Sigma, \delta)$, $S \subseteq Q$, and a relation $R \subseteq Q^2$. \\
  \textbf{Output:} Yes, if and only if there exists a word $w \in \Sigma^*$ with $|S.w| = 1$ and $R \subseteq \lessdot_w$.
\end{problem*}
}

\comment{
It is reasonable to distinguish whether the order should include the initial configuration of the automaton or if it should only describe the consequences of the chosen transitions. 
In the former case, we refer to the problem as {\sc Sync-Under-$\mathit{0}$-$\lessdot_w$} (starting at $w[0]$), in the latter case as {\sc Sync-Under-$\mathit{1}$-$\lessdot_w$} (starting at $w[1]$), and if the result holds for both variants, we simply refer to is as {\sc Sync-Under-$\lessdot_w$}.
Examples for positive and negative instances of the problem synchronization under order for some discussed variants are illustrated in Figure~\ref{fig:expl}. 
}

 Given a DFA $\Aa = (Q, \Sigma, \delta)$, define $\first(q, w, S)$ to be the
minimum of the positions at which the state $q$ appears as an active state over all paths induced by $w$ starting at some state in $S$. Accordingly, let $\last(q, w, S)$ be the maximum of those positions.
 Note that $\first(q, w, S) = 0$ for all states $q \in S$ and is strictly positive for $q \in Q \backslash S$. If~$q$ does not
 appear on a path induced by $w$ on $S$,
we set $\first(q, w, S) := +\infty$ and $\last(q, w, S) := -\infty$. For any natural number $n$, we make standard assumptions that $-\infty < n < +\infty$, and we also assume that $-\infty \le -\infty$, $+\infty \le + \infty$, but the inequalities $-\infty < -\infty$ or $+\infty < +\infty$ do not hold.

We remark that in \cite{Wolf2020}, the traversal relations $\mathrm{R}(w)$ (denoted $\lessdot_w$ there) were called orders. However, in some cases they are not transitive or reflective, and are thus not order relations. Instead, these traversal relations describe the order in which $w$ traverses the DFA. The traversal relations $\ore(w)$ and $\orz(w)$ defined below were denoted as $\propto^{l<l}_{w @s}$ and $\propto^{l\le l}_{w @s}$, and the relations $\orep(w)$, $\orzp(w)$ and $\ordp(w)$ were denoted there as $\propto^{l<l}_{w @p}$, $\propto^{l\le l}_{w @p}$ and $\propto^{l<f}_{w @p}$. Finally, we mention that in \cite{Wolf2020} the case where the configuration before reading the first letter of the word is not included in the definitions of $\first(q, w, S)$ and $\last(q, w, S)$ was also considered. However, this variant does not change the complexity of the problems that we consider, and hence we omit its discussion.

In the following definitions, let $\Aa = (Q, \Sigma, \delta)$ be a DFA and let $p, q \in Q$ be two its distinct states. For all the traversal relations defined below we do not include the pair $(p, p)$ for any $p \in Q$.

\paragraph*{Traversal relations on sets.}  Given a word $w \in \Sigma^*$, we define the following traversal relations~$\mathrm{R}(w)$:
	$$(p, q) \in \ore(w) \Leftrightarrow
	\last(p, w, Q) < \last(q, w, Q) \textmd{ (traversal relation $\ell < \ell$ on sets),}$$
	 $$(p, q) \in  \orz(w) \Leftrightarrow 
	\last(p, w, Q) \leq \last(q, w, Q) \textmd{ (traversal relation $\ell \le \ell$ on sets).}$$

The second traversal relation differs from the first one in the sense that it does not say that $q$ is active strictly after the moment when $p$ is active for the last time, and instead they can disappear simultaneously. It is easy to see that for any word $w$ we have~$\ore(w) \subseteq \orz(w)$.

\paragraph*{Traversal relations on paths.} So far, we only introduced traversal relations which are defined by the set of active states as a whole. It did not matter which active state belongs to which path and a state on a path $\tau$ could stand in a relation with a state on some other path $\rho$. However, in most scenarios the fact that we start with the active state set $Q$ only models the lack of knowledge about the \emph{actual} current state, and in reality only one state $q$ is active at any moment of time. Hence, any constraints on the ordering of traversed states should apply to each path separately. Therefore, we introduce the variants of the traversal relations above which are defined on paths rather than on series of state sets:
$$(p, q) \in \orep(w) \Leftrightarrow 
	\forall r \in Q \colon 
	\last(p, w, \{r\}) < \last(q, w, \{r\}) \textmd{ (traversal relation $\ell < \ell$ on paths),}$$
$$(p, q) \in \orzp(w) \Leftrightarrow 
	\forall r \in Q \colon 
	\last(p, w, \{r\}) \leq \last(q, w, \{r\}) \textmd{ (traversal relation $\ell \le \ell$ on paths).}$$

The traversal relations $\orep$ and $\orzp$ differ significantly. For instance, for any pair $(p, q) \in \orep(w)$ every path labeled by $w$ visits $q$, while for $(p, q) \in \orzp(w)$ a path labeled by $w$ may not necessarily visit $q$ if it does not visit $p$ as well. However, for any word $w$ we have~$\orep(w) \subseteq \orzp(w)$.

In a way, restricting the traversal relations described so far with a pair $(p, q)$ in the constraint relation $R$ brings ``positive'' constraints on the future transitions of a word, in the sense that the visit of a state $p$ demands for a later visit of the state $q$. For instance, opening the lid demands closing the lid later in our example in the introduction. We now introduce a traversal relation whose restrictions yield ``negative'' constraints. The last kind of traversal relations demands for a pair of states $(p, q)$ that the first visit of the state $q$ forbids any future visits of the state $p$. Referring again to the example in the introduction, this might be caused by the requirement to not turn the box after opening the lid. This stands in contrast to the previous restrictions where we could make up for a ``forbidden'' visit of the state $p$ by visiting $q$ again.
The traversal relation $\mathrm{R}^{\ell < f}$ will only be considered on paths, since the corresponding definition on sets leads to the empty relation alreafy on the state set $Q$. The definitions of $\ordp$ is as follows:
$$(p, q) \in \ordp(w) \Leftrightarrow \forall r \in Q:
	\last(p, w, \{r\}) < \first(q, w, \{r\}) \textmd{ (traversal relation $\ell < f$ on paths).}$$

Note that $\ordp$ is not transitive. For example, for $R=\{(p,q), (q,r)\}$ a path is allowed to go from $r$ to~$p$ if it has not visited $q$.

There are very few connections between different introduced traversal relations. For every DFA and every word $w$, besides the already mentioned facts that $\ore(w) \subseteq \orz(w)$ and $\orep(w) \subseteq \orzp(w)$, we also have that $\ordp(w) \subseteq \orzp(w)$. No other inclusions hold true for every DFA and every $w$. \Cref{fig-example-relations} provides some examples on the differences between the behavior of various traversal relations.

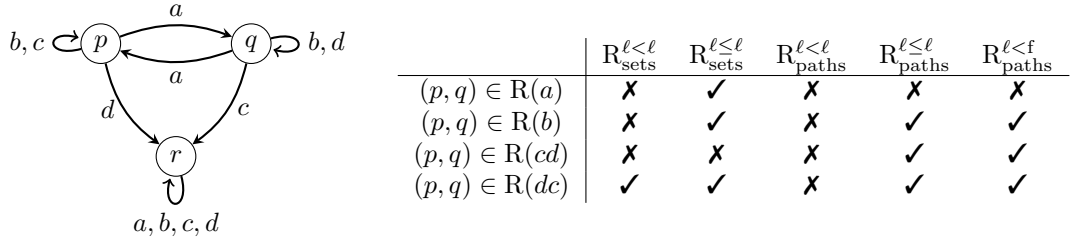
\begin{figure}[ht]\centering
  \begin{minipage}{.25\textwidth}\centering

	\begin{tikzpicture} [node distance = 2cm]
	\tikzset{every state/.style={inner sep=1pt,minimum size=1.5em}}
	
	\node [state] at (0,0) (p) {$p$};
 \node [state] at (2,0) (q) {$q$};
  \node [state] at (1,-1.5) (r) {$r$};
	\path [-stealth, thick]
 (p) edge [bend left=20] node[above] {$a$} (q)
  (q) edge [bend left=20] node[below] {$a$} (p)
 (p) edge [loop left] node {$b, c$} (p)
 (q) edge [loop right] node {$b, d$} (q)
  (p) edge [bend right=20] node[left] {$d$} (r)
  (q) edge [bend left=20] node[right] {$c$} (r)
  (r) edge [loop below] node {$a, b, c, d$} (r)

	;
	\end{tikzpicture}
   \end{minipage}%
  \begin{minipage}{.75\textwidth}
      \centering
      \begin{tabular}{c|c c c c c}

            & $\ore$ & $\orz$ & $\orep$ & $\orzp$ & $\ordp$ \\ \hline
           $(p, q) \in \mathrm{R}(a)$ & \ding{55} & \ding{51} & \ding{55} & \ding{55} & \ding{55} \\ 
           $(p, q) \in \mathrm{R}(b)$ & \ding{55} & \ding{51} & \ding{55} & \ding{51} & \ding{51} \\  
           $(p, q) \in \mathrm{R}(cd)$ & \ding{55} & \ding{55} & \ding{55} & \ding{51} & \ding{51} \\  
           $(p, q) \in \mathrm{R}(dc)$ & \ding{51} & \ding{51} & \ding{55} & \ding{51} & \ding{51} \\  
      \end{tabular}
  \end{minipage}
	\caption{Example of a DFA and a table showing which traversal relations include the pair~$(p, q).$}\label{fig-example-relations}
\end{figure}

\paragraph*{Connection to \LTL{}.} Take the set of states of $\Aa$ as the set of atomic propositions. Then, for traversal relations on paths, each state on a path is labeled only by itself, while for traversal relations on sets each state of $\PP(\Aa)$ is labeled by all the states of $\Aa$ that it includes. Then, all defined orders can be expressed in \LTL{}. 

Indeed, let $R$ be the constraining relation, and let $(p, q) \in R$. Then for $\mathrm{R}^{l \leq l}$ we take the formula $\Glob (p \to \Event q)$, since each occurrence of $p$ requires an occurrence of $q$ in the future. For $\mathrm{R}^{l < l}$ we take $\Event(q \land \Glob \neg p)$, since $q$ must occur at some point, and after the last of its occurrence, there must not be any occurrence of $p$. Finally, for $\mathrm{R}^{l < f}$ we take $ \Glob (q \to \Glob \neg p)$, since after the first occurrence of $q$ there must be no occurrence of $p$. To combine the constraints induced by multiple pairs in $R$, we take the conjunction of the corresponding formulas.

To express the fact that a word is synchronizing in \LTL{}, one can use the following approaches. Pick any state $s$ in a sink maximal strongly connected component of the DFA. On paths, it is enough to add the restriction $\Event \Glob s$.
On sets, we can introduce an atomic proposition that is true for all the states except $s$, and put $\Event \neg s$ as a restriction. We do not use either of these approaches, and instead add a direct separate requirement of synchronizability to the problems, since this allows us to get more precise results on the shape of formulas and the set of atomic propositions. For all our complexity lower bounds we restrict to the setting where the only atomic proposition that is true in a state is the name of the state itself, thus strengthening the lower bounds in \cite{Bertrand2023}, where the same atomic proposition is allowed to be true in a set of states of non-constant length.

%% file: sec-pspace.tex
\paragraph*{Generic model checking algorithm.} In this section, we describe polynomial space algorithms for \textsc{\LTL{} model checking on paths (sets)} and their synchronized versions. \textsc{Synchronized \LTL{} model checking on paths} was proved to be in \PSPACE{} in \cite{Bertrand2023}. While most of the ideas are already present there, we need several tweaks of their polynomial space algorithm, hence we choose to provide its brief description in a way that is more convenient for our use. Then, we describe how to adapt it to the three other problems defined in \Cref{subsec:\LTL{}-defs}: \textsc{\LTL{} model checking on sets}, \textsc{\LTL{} model checking on paths}, and
\textsc{Synchronized \LTL{} model checking on sets}. The main issue that the algorithm needs to overcome is that the \LTL{} restrictions are defined on the product DFA (on paths) or the power-set DFA (on sets), which are both of exponential size in terms of the input. However, as we discuss, we do not have to construct them explicitly, and only need to store one their state at a time. 

The idea of the approach in \cite{Bertrand2023} is an extension of the classical automata-based polynomial space algorithm for \LTL{} model checking. The main difference is that the automaton defining the paths now has exponential size, and hence has to be dealt with without constructing it. We provide only a brief description of the algorithm of \cite{Bertrand2023}, and we refer to the proof of Theorem~3 there for full details. 
We work in the following slightly more general setting. Let $P$ be a set of atomic propositions and let $\CC = (Q_\CC, \Sigma, \delta_\CC, i, F)$ be a DFA acceptor with a labelling function $L: Q_\CC \to 2^P$. We think of $\CC$ and $L$ as defined implicitly by a much smaller DFA $\Aa$ provided in the input, the precise constructions will be specified later. The requirements for this implicit specification are that the transition function of $\CC$ and the labeling function $L$ must be computable in polynomial space in the size of $\Aa$. In particular, this means that every state of $\CC$ and the set of labels must be of polynomial size in the size of $\Aa$. 

Given $\CC$ defined in such a way, and an \LTL{} formula $\phi$ over the set $P$ of atomic propositions, we want to decide if there exists a word accepted by $\CC$ and satisfying $\phi$. Given $\CC$ and $\phi$, one can construct a non-deterministic finite automaton $\DD$ with the following property: the language accepted by $\DD$ is non-empty if and only if there exists a word $w \in \Sigma^*$ accepted by $\CC$ such that the path labeled by it in~$\CC$ satisfies $\phi$. To build $\DD$, we first construct the automaton $\BB_\phi$ over the alphabet $2^P$ which recognizes the set of finite traces of $\CC$ satisfying $\phi$, and take the synchronized product of $\CC$ and $\BB_\phi$ as $\DD$, please refer to the proof of Theorem 3 of \cite{Bertrand2023} for technical details. We only need the fact that each state of $\BB_\phi$ is of size polynomial in $\phi$, and its transition function is computable in polynomial time. If $\CC$ satisfies the restrictions described in the previous paragraph, then, by the construction of synchronized product, each state of $\DD$ also has polynomial size and the transition function of $\DD$ is computable in polynomial space.  To check that the language accepted by $\DD$ is non-empty in non-deterministic polynomial space, it is enough to non-deterministically guess a word $w$ accepted by $\DD$ letter by letter. Observe that for each guess we do not need the whole automaton $\DD$ to be constructed explicitly, since we only need to maintain the current state and then compute one of the states obtained by taking a transition labeled by a guessed letter. By Savitch theorem \cite{DBLP:books/daglib/0086373}, \PSPACE{}$=$\NPSPACE{}, and hence this can be done in deterministic polynomial space. 

\paragraph*{Implicitly represented DFAs.} Let now $\Aa = (Q, \Sigma, \delta)$ and $\phi$ be the input of the \textsc{Synchronized \LTL{} model checking on paths} problem, and let $|Q| = n$. Take $\CC$ to be the product of $n$ copies of~$\Aa$. As the initial state of $\CC$, we take an arbitrary $n$-tuple with pairwise different states. The set of accepting states of $\CC$ consists of all the $n$-tuples with the same state at each position. Such automaton recognizes the set of synchronizing words of $\Aa$, and its transition function is computable in polynomial time. As the set of labels, we take a separate label $q_i^{(j)}$ for each state $q_i$ of the $j$th copy of $\Aa$. The labeling function $L$ is then defined as  $L((q_{i_1}, \ldots, q_{i_n})) = \{q_{i_1}^{(1)}, \ldots, q_{i_n}^{(n)}\}$. Finally, for each $j$, define $\phi_j$ to be $\phi$ where each state $q_i$ is substituted with $q_i^{(j)}$, and let $\phi'$ be the conjunction of all $\phi_j$. It is then easy to see that a word $w \in \Sigma^*$ accepted by $\CC$ labels in it a path satisfying $\phi'$ if and only if for every $\pi \in \paths(Q, w)$ we have $\pi \models \phi$. By using the algorithms described above, we get that \textsc{Synchronized model checking on path} is in \PSPACE{}, which was originally proved in \cite{Bertrand2023}.

We now discuss the modifications for other problems. First, we observe that if we make all the states of $\CC$ accepting in the above construction, we get a polynomial space algorithm for \textsc{\LTL{} model checking on paths}.

\begin{proposition}\label{prop:mc-paths-inpspace}
\textsc{\LTL{} model checking on paths} and its synchronized version are in \PSPACE{}.
\end{proposition}

Now, take $\CC$ to be the power-set automaton $\mathcal{P}(\Aa)$. Take the set $Q$ of all the states of $\Aa$ as the initial state of $\mathcal{P}(\Aa)$, take $P = Q$ and define the labeling function $L$ as $L(\{q_1, \ldots, q_k\}) = \{q_1, \ldots, q_k\}$ for $q_1, \ldots, q_k \in Q$. By defining the set of accepting states of $\mathcal{P}(\Aa)$ as, respectively, the set $2^{Q}$ of all its states, or the set of singletons, we get the following result.

\begin{proposition}\label{prop:mc-sets-inpspace}
\textsc{\LTL{} model checking on sets} and its synchronized version are in \PSPACE.
\end{proposition}

In particular, as explained in \Cref{subsec-order-defs}, the constraints on all considered traversal relations are expressible as \LTL{} formulas. Hence, we get a uniform proof of Theorem 11 from \cite{Wolf2020}. Moreover, this theorem is true for arbitrary Boolean combination of constraints on different traversal relations.

\begin{theorem}[\cite{Wolf2020}]
	For all traversal relations $\mathrm{R}(w) \in \{\ore, \orep, \orz, \orzp, \ordp\}$, the problem \textsc{Constrained-$\mathrm{R}(w)$-Sync} is contained in \PSPACE.
\end{theorem}

Moreover, the described algorithm can be easily adapted to require the word we are looking for to have given rank, or bring the whole set of states to a particular subset (or a family of subsets), or even come from a given fixed regular language. It is also easy to see that all these restrictions can be implemented in the setting where we are also given a subset $S \subseteq Q$ in the input, and we need to check if there exists a word $w$ that for every $\pi \in \paths(S, w)$ we have $\pi \models \phi$ (respectively, for the path $\pi = \sets(S, w)$ we have $\pi \models \phi$). This word $w$ can be additionally required to synchronize $S$, that is, to have $|S.w| = 1$. Finally, we remark that even if $\Aa$ is not deterministic, without any changes, the described polynomial space algorithm still finds a word satisfying the given \LTL{} constraints on paths or on sets.

%% file: sec-subsets.tex
In this section, we give two constructions which provide \PSPACE{}-hardness for most of the constraints that we introduced, one for traversal relations on sets and one for traversal relations on paths. As discussed at the end of \Cref{subsec-order-defs}, these constraints can be expressed by \LTL{} formulas, hence as a consequence we get \PSPACE{}-hardness of model checking on path and sets, as well as of their synchronized versions.

\subsection{Constraints on sets}

We start with a reduction from the \PSPACE{}-complete problem \textsc{Careful Sync} \cite{DBLP:conf/wia/Martyugin12,DBLP:journals/mst/Martyugin14}  to \CS{\orz}, and then extend the construction to other traversal constraints.

\begin{problem*}\label{prob-subset-sync-under}
  \textsc{Careful Sync} \\
  \textbf{Input:} A partial DFA $\Aa=(Q, \Sigma, \delta)$. \\
  \textbf{Output:} Yes, if and only if there exists a word $w\in \Sigma^*$, such that $w$ is defined for all $q\in Q$ and $|Q.w| = 1$.
\end{problem*}

\begin{theorem}
	\label{thm:nonstr-on-sets}
	\CS{\orz} and \CS{\ore} are \PSPACE-complete, even for $|R| = 1$ and $|\Sigma| = 2$. \CW{\ore} is \PSPACE-complete, even for $|R| = 2$ and $|\Sigma| = 3$.
\end{theorem}
\begin{proof}
Membership in \PSPACE{} comes from \Cref{prop:mc-sets-inpspace} and the fact that constraints on $\orz$ and $\ore$ can be expressed in \LTL{}. We first show \PSPACE-hardness of \CS{\orz}. Let $\Aa=(Q, \Sigma, \delta)$ be a partial DFA in the input of \textsc{Careful Sync}. We construct a complete DFA $\Aa'= (Q', \Sigma, \delta')$ with $Q' = Q\cup \{q_{\circleddash}, r\}$ such that $q_{\circleddash}, r \notin Q$. We define the constraint relation $R$ as $R = \{(q_{\circleddash}, r)\}$. The idea is to define the transition function of $\Aa'$ so that after reading any non-empty word, it is no longer possible to visit the state $r$, and since $(q_{\circleddash}, r) \in R$, it is also not possible to visit~ $q_{\circleddash}$. We then define all the transitions undefined in $\Aa$ to go to $q_{\circleddash}$ in $\Aa'$. Thus we get that for every word $w$ we have that $\orz(w)$ agrees with $R$ and $w$ is synchronizing for the $\Aa'$ if and only if $w$ is carefully synchronizing for the original automaton $\Aa$. See \Cref{fig-sets} for an illustration.

Formally, the transition function $\delta'$ is defined as follows. We take $\delta'(q, a) = \delta(q, a)$ for all the pairs $q \in Q, a \in \Sigma$ where $\delta$ is defined. For those pairs $(q, a)$ where $\delta$ is not defined, we set $\delta'(q, a) = q_{\circleddash}$.
	Further, fix some arbitrary state $t\in Q$ and for all $a \in \Sigma$ set $\delta'(q_{\circleddash}, a) = \delta'(t, a)$ (note that this can be the state $q_{\circleddash}$ itself) and $\delta'(r, a) = \delta'(t, a)$.
	
\begin{figure}[ht]\centering
	\begin{tikzpicture} [node distance = 2cm]
	\tikzset{every state/.style={inner sep=1pt,minimum size=1.5em}}
	
	\node [state] at (4,0) (q2) {};
 \node [state] at (4,-1) (q3) {};
 \node [state] at (4,1) (t) {$t$};
 
 \node [state, dotted] at (5,2) (ht) {$\hat{t}$};
 
 \node [state] at (3,0) (s) {$s$};
 \node [state] at (3,-1) (ns) {};

 \node [state, dashed] at (0.5,0) (qp1) {$q_\oplus^1$};
 \node [state, dashed] at (-1.3,0) (qp2) {$q_\oplus^2$};

	\node [state] at (7.3,-1) (qc) {$q_{\circleddash}$};
	\node [state] at (7.3, 1) (r) {$r$};

 \node [state, dotted] at (8.3, 2) (hr) {$\hat{r}$};
	
	\draw (3.5,0) ellipse (1.5cm and 1.7cm);
	
	\path [-stealth, thick]
 (q3) edge [bend right=30] node[below] {} (qc)
  (q2) edge [bend left=20] node[below] {} (qc)
	(t) edge node[left] {$b$} (q2)
 	(qc) edge[bend left=20] node[below] {$b$} (q2)
  (r) edge[bend right=20] node[above] {$b$} (q2)
  (t) edge[bend left=20] node[above] {$a$} (qc)
 	(qc) edge [loop below] node {$a$} (qc)
  (r) edge node[right] {$a$} (qc)

  (s) edge[dashed] node[below] {$c$} (qp1)
  (qp1) edge[dashed] node[below] {$c$} (qp2)
  (qp1) edge[dashed, loop below] node[below] {$a, b$} (qp1)
  (qp2) edge[dashed, loop below] node[below] {$a, b, c$} (qp2)

  (ns) edge[bend right=40, dashed] node[below] {$c$} (qc)
  (r) edge[bend left=60, dashed] node[right] {$c$} (qc)

  (qc) edge [loop right, dashed] node {$c$} (qc)

  (hr) edge[dotted] node[right] {} (r)
  (ht) edge[dotted] node[right] {} (t)
	;
	\end{tikzpicture}
	\caption{Illustration of the constructions in the proof of \Cref{thm:nonstr-on-sets}. The construction for $\orz$ is depicted by solid lines, its extension to $\ore$ by dotted lines, and the further extension to drop the synchronizability requirement by dashed lines. For every state $q$ except $q_\circleddash$, we have a dotted state~$\hat{q}$, some of them are omitted in the picture for simplicity. Unmarked states and transitions are only drawn for illustrative purposes and their labels do not matter for the proof.}
 \label{fig-sets}
\end{figure}
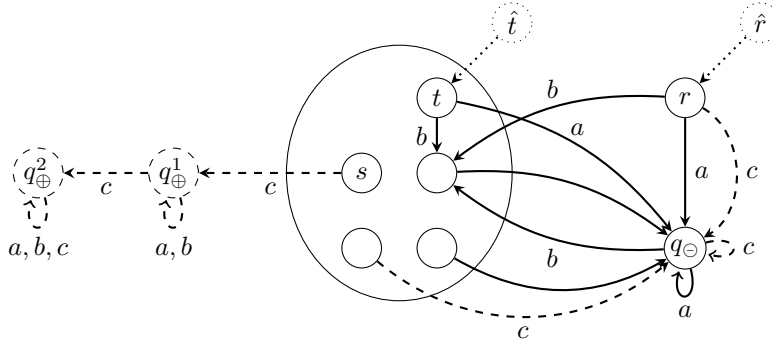

Assume that there exists a word $w \in \Sigma^*$, $|w| = n$, such that for all states $q \in Q$,  $\delta(q, w)$ is defined, and $|\delta(Q, w)| = 1$.
	In particular, $\delta(q, w[1])$ is defined for all states $q \in Q$.
Let us consider how the letter $w[1]$ acts on the states of $\Aa'$. First,  $\delta'(r, w[1]) = \delta'(q_\circleddash, w[1]) = \delta(t, w[1])$, and $\delta(t, w[1])$ is defined by our assumption. Second, $\delta'(Q, w[1]) \subseteq Q$ since $\delta(q, w[1])$ is defined for all states $q \in Q$. Hence, $\delta'(Q', w[1]) \subseteq Q$. Furthermore, by construction of $\delta'$, we have $\delta'(Q', w[1]) = \delta(Q, w[1])$. Let $w = w[1]w'$. 
	By our assumption, $\delta(q, w')$ is defined for every $q \in \delta(Q, w[1])$, and by construction $\delta'(q, w') = \delta'(q, w')$ for every $q \in \delta(Q, w[1])$. In particular, this means that while reading $w'$ in $\Aa'$ starting from the states in $\delta'(Q', w[1])$, 
	the state $q_\circleddash$ is not visited,
	and $\delta'(Q', w) = \delta(Q, w)$. 
	Therefore,~$w$ also synchronizes the automaton~$\Aa'$. 
	The state $q_\circleddash$ is 
	only active in the start configuration where no letter of $w$ is read yet, and after that it is not active anymore while reading $w$. The same is true for~$r$, hence $R = \{(q_\circleddash, r)\} \subseteq\ \orz(w)$. 
	
	In the opposite direction, assume there exists a word $w \in \Sigma^*$, $|w| = n$, that synchronizes~$\Aa'$ with $(q_\circleddash, r) \in \orz(w)$.
	The only position of $w$ in which $r$ is active in $\Aa'$ due to the definition of $\delta'$ is before any letter of $w$ is read. 
	As $(q_\circleddash, r)\in\ \orz(w)$, it holds for all $i$, $1 \le i \le n$, that $q_\circleddash \notin\delta'(Q', w[1..i])$. Hence, $\delta'(q, w)$ is defined for every state $q \in Q$. Since $\delta'$ and $\delta$ agree on the definition range of~$\delta$, it follows that $w$ also synchronizes the state set $Q$ in $\Aa$ without using an undefined transition.

 To adapt the construction for $\ore$, we 
	introduce a copy $\hat{q}$ of every state $q \in Q \cup \{r\}$ and set $\delta'(\hat{q}, a) = q$ for every $a \in \Sigma$, $q \in Q\cup \{r\}$. 
	We keep $R = \{(q_\circleddash, r)\}$. Since for any word $w\in \Sigma^*$ with $|w|\geq 2$, we have that $r$ is no longer active after reading $w[2]$, in order for $\ore(w)$ to agree with $R$,
	the state~$q_\circleddash$
	needs to never be active after reading $w[1]$ in $w$. Note that the state $q_\circleddash$ was not copied.

Denote the modified DFA (used for $\ore$) as $\Aa'$ again. To show that \CW{\ore} is \PSPACE{}-complete, we further modify this construction so that the requirement on the word to be synchronizing for $\Aa'$ can be dropped. Add a new letter $c$ to the alphabet $\Sigma$, and add two new states $q^1_{\oplus}, q^2_{\oplus}$ to $Q'$. Observe that if $\Aa$ is carefully synchronizing, then every carefully synchronizing word maps all the states of $\Aa$ to the same maximal strongly connected component of the underlying digraph of $\Aa$. More generally, if $\Aa$ is carefully synchronizing, then the underlying digraph of $\Aa$ can only contain one maximal strongly connected component $C\subseteq Q$ that is a sink. For the sake of contradiction, assume that the underlying digraph of $\Aa$ contains two maximal strongly connected components $C_1$ and $C_2$ that are both sinks. Then, no word can synchronize the set $C_1 \cup C_2$ as no state outside of $C_1 \cup C_2$ is reachable from any state in $C_1 \cup C_2$, and the states in $C_1$ cannot be mapped into $C_2$ and vice versa. Observe that for a digraph, all the maximal strongly connected components that are sinks can be found in polynomial time.

Let $C$ be the single maximal strongly connected component of $\Aa$ that is a sink.
Pick a state $s$ in this component. Define all the letters in $\Sigma$ to act as the identity on $q^1_{\oplus}$ and $q^2_{\oplus}$. Furthermore, define $\delta'(s, c) = q^1_{\oplus}$, $\delta'(q^1_{\oplus}, c) = q^2_{\oplus}$, $\delta'(q^2_{\oplus}, c) = q^2_{\oplus}$ and $\delta'(q, c) = q_\circleddash$ for all $q \in Q' \setminus \{s, q^1_{\oplus}, q^2_{\oplus}\}$. Define $R' = R \cup\{(q^1_{\oplus}, q^2_{\oplus})\}$. As shown above, if $R'$ is respected, the state $q_\circleddash$ is never active after reading the first letter of any word.
Observe that since $(r, q_\circleddash) \in R'$, letter $c$ cannot be applied if any state in $Q'$ other than $s$ is active. However, it has to be applied since $(q^1_{\oplus}, q^2_{\oplus}) \in R'$, and applying any letter in $\Sigma$ preserves the fact that both $q^1_{\oplus}$ and $q^2_{\oplus}$ are active. Hence, for a word $w \in \Sigma^*$, $\ore(wcc)$ agrees with $R'$ if and only if $w$ is synchronizing for $\Aa'$ and $\ore(wcc)$ agrees with $R$. Hence, we proved that \CW{\ore} is \PSPACE{}-complete.
\end{proof}

Together with \Cref{prop:mc-sets-inpspace}, we get the following.

\begin{theorem}\label{thm:mc-sets-main}
\textsc{Model checking on sets} and its synchronized version are \PSPACE-complete, even for a fixed \LTL{} formula and a DFA over constant-size alphabet.
\end{theorem}

An example of a fixed formula for which \textsc{Model checking on sets} is \PSPACE{}-complete is $\Glob (q_1 \to \Event q_2)$. For \textsc{Synchronized model checking on sets} an example of such a formula is $\Glob (q_1 \to \Event q_2) \land \Glob (q_3 \to \Event q_4)$.

\subsection{Constraints on paths} \label{subsec:pspace-path}

Observe that $\orzp$ and $\ordp$ express similar properties: in both cases, if $(p, q) \in R$, then every path visiting $p$ must visit $q$ later. This allows us to use almost the same construction to show \PSPACE{}-hardness for both of them. For \CS{\ordp}, this was already proved in \cite{Bertrand2023}, but the alphabet in their construction grows with the number of states of the DFA. We show, in particular, that for constant-size alphabet this problem remains \PSPACE{}-hard.

\begin{theorem}\label{thm:pspace-on-paths}\CS{\orzp}, \CS{\ordp} and \CW{\orzp} are \PSPACE{}-complete, even over a constant-size alphabet.
\end{theorem}

\begin{proof}
Membership in \PSPACE{} comes from \Cref{prop:mc-paths-inpspace} and the fact that constraints on $\orzp$ and $\ordp$ can be expressed in \LTL{}. To show \PSPACE-hardness, we reduce from the \textsc{Finite Automata Intersection} problem, which is \PSPACE-complete, even for binary alphabets~\cite{Kozen1977}.

\begin{problem*}\label{prob-automata-intersection}
  \textsc{Finite Automata Intersection} \\
  \textbf{Input:} A set of $m$ DFA acceptors $\Aa_i = (Q_i, \Sigma, \delta_i, s_i, F_i), 1 \le i \le m$. \\
  \textbf{Output:} Yes, if and only if there exists a word in $\Sigma^*$ which is accepted by all  $\Aa_i, 1 \le i \le m$.
\end{problem*}

The idea is as follows. We construct a single DFA which is a union of all the DFA acceptors together with three new states $y, n, f$. We add a letter $r$ which makes only the initial state of each DFA acceptor active. Then, the synchronizability requirement will force all these initial states to be brought to the same state $f$ by a new letter $t$, and the constraints on traversal relations will force them to go through the state $y$ instead of $n$ on the way to $f$, thus guaranteeing that every synchronizing word satisfying the constraints corresponds to a word accepted by every DFA acceptor.

Formally, given an instance of \textsc{Finite Automata Intersection}, we construct the following DFA $\BB = (Q', \Sigma \cup \{r, t\}, \delta')$. We take $Q' = \cup_{i = 1}^m Q_i \cup \{y, n, f\}$. For each $1 \le i \le m$, we define $\delta'(q, a) = \delta_i(q, a)$ for all $q \in Q_i$, $a \in \Sigma$. For the states $y, n, f$ we set every letter of $\Sigma$ to induce a self-loop. It remains to define $\delta'$ for the new letters $r$ and $t$. For each $1 \le i \le m$, the letter $r$ sends all the states in $Q_i$ to $s_i$, and the letter $t$ takes all the states in $F_i$ to $y$, all the states in $Q_i \setminus F_i$ to $n$, and sends the states $n$ and $y$ to $f$. All yet unspecified transitions induce self-loops. The part of the constructed DFA corresponding to $\Aa_1$ and $\Aa_2$ is depicted in Figure \ref{fig-fai} (note that the states $y, n, f$ are shared by all such gadgets).
	
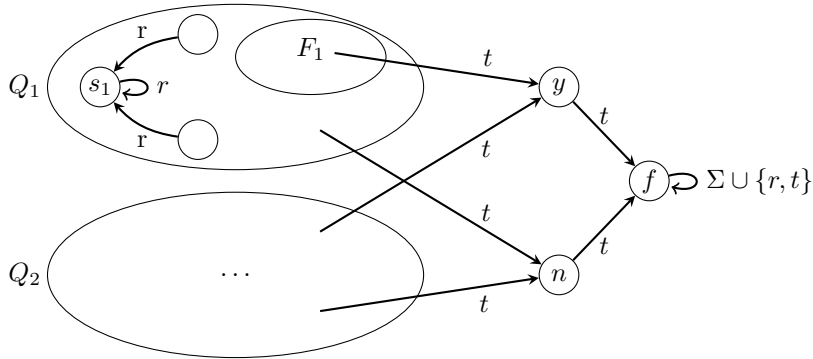
\begin{figure}[ht]\centering
	\begin{tikzpicture} [node distance = 2cm]
	\tikzset{every state/.style={inner sep=1pt,minimum size=1.5em}}
	
	\node [state] at (3,0.7) (q1) {};
	\node [state] at (3,-0.7) (q2) {};
	\node [state] at (1.7,0) (s1) {$s_1$};
	
	\node [state] at (7.8,-2.5) (n1) {$n$};
	\node [state] at (7.8, 0) (y1) {$y$};
	
	\node [state] at (9, -1.25) (f) {$f$};
	
	\draw (3.5,0) ellipse (2.5cm and 1.1cm);
	\draw (4.5,0.4) ellipse (1cm and 0.5cm);
	\node [] at (4.5,0.5) (F1) {$F_1$};
	\node [] at (4.5,-0.5) (NF1) {};
	\node [] at (0.7,0) {$Q_1$};

    \draw (3.5,-2.5) ellipse (2.5cm and 1.1cm);
    \node [] at (3.5,-2.5) {$\ldots$};
    \node [] at (0.7,-2.5) {$Q_2$};
    \node [] at (4.5,-2) (F2) {};
    \node [] at (4.5,-3) (NF2) {};
    
    \path [-stealth, thick]
	(q1) edge [bend right=20] node[above] {r} (s1)
	(q2) edge [bend left=20] node[below] {r} (s1)
	(F1) edge node[above, near end] {$t$} (y1)
	(NF1) edge node[above, near end] {$t$} (n1)
	(n1) edge node[below] {$t$} (f)
	(y1) edge node[above] {$t$} (f)
	(f) edge [loop right] node {$\Sigma \cup \{r, t\}$} (f)
	(s1) edge [loop right] node {$r$} (s1)

 	(F2) edge node[below, near end] {$t$} (y1)
	(NF2) edge node[below, near end] {$t$} (n1)

	;
	\end{tikzpicture}
	\caption{Illustration of the construction in the proof of \Cref{thm:pspace-on-paths}. }\label{fig-fai}
\end{figure}

 Finally, we take $R_1 = \{(s_i, y) \mid 1 \le i \le m\}$ and $R_2 = \{(n, s_i) \mid 1 \le i \le m\}$. We claim that there exists a synchronizing word $w$ for $\BB$ such that $\orzp(w)$ (respectively, $\ordp(w)$) agrees with $R_1$ (respectively, $R_2$) if and only if all the DFAs $\Aa_1, \ldots, \Aa_m$ accept the same word.
 
	In one direction, if there is a word $w$ accepted by $\Aa_1, \ldots, \Aa_m$, then $rwtt$ is a synchronizing word for $\BB$ such that $R_1 \subseteq \orzp(rwtt)$ and $R_2 \subseteq \ordp(rwtt)$. 
	Indeed, since for each $1 \le i \le m$ the word $w$ is accepted by $\Aa_i$, we have that $\delta(q, rwt) = y$ for each $q \in Q_i$. After that, the application of $t$ sends $y$ to $f$ without visiting $n$. Hence, for each $1 \le i \le m$ the path labeled by $rwtt$ and staring in~$s_i$ visits $y$ and then ends in $f$, so it visits $y$ for the last time after its last visit of $s_i$. Hence $rwtt$ is a synchronizing word such that $\orzp(rwtt)$ agrees with $R_1$. Moreover, such path never visits $n$ and hence $\ordp(rwtt)$ agrees with $R_2$.

In the other direction, assume that $w$ is a synchronizing word for $\BB$ such that $\orzp(w)$ (respectively, $\ordp(w)$) agrees with~$R_1$ (respectively, with $R_2$). By construction of $\BB$, since $w$ is synchronizing, for each $1 \le i \le m$ the word $w$ sends $s_i$ to $f$, as $f$ is  a sink.
Observe that the word $w$ must contain at least one occurrence of $t$. Take $w = w_1tw_2$, where $w_1$ does not contain any occurrences of $t$. If $w_1$ contains at least one occurrence of $r$, let $w = w_3rw_4$, where $w_4$ does not contain any occurrences of~$r$. Otherwise, take $w_4 = w_1$. Since $\orzp(w)$ (respectively, $\ordp(w)$) agrees with $R_1$ (respectively, $R_2$), for each $1 \le i \le m$ the word $w_4$ must map $s_i$ to $y$, because otherwise the path labelled by $w$ starting at $s_i$ would visit $n$ instead of $y$, and the states $y$ and $n$ cannot be visited by the same path in $\BB$. Hence, the word $w_4$ maps the initial state of each $\Aa_i$, $1 \le i \le m$, to its final state, and thus it is accepted by each DFA acceptor.

The constraint relation $R_1$ can be extended to force $w$ to be synchronizing. Indeed, add the pairs $(p, f)$ for all states $p \in Q' \setminus \{f\}$ to $R_1$. It is easy to see that in this case every word $w$ such that $\orzp(w)$ agrees with $R_1$ is synchronizing. Hence, \CW{\orzp} is also \PSPACE{}-complete.
\end{proof}	

Finally, the same arguments as in the proof of \Cref{thm:pspace-on-paths} show that satisfying the \LTL{} formula $\Event f \land ((\neg n \land \neg f) \to \Event y)$
is equivalent to the requirement for a word to be accepted by all the DFA acceptors. This is also true for every synchronizing word satisfying $\neg n \to \Glob \neg n$. Hence, together with \Cref{prop:mc-paths-inpspace}, we get the following.

\begin{theorem}\label{thm:mc-paths-main}
\textsc{\LTL{} model checking on paths} and its synchronized version are \PSPACE-complete, even for a fixed $\LLL^+(\Event)$ \LTL{} formula. \textsc{Synchronized \LTL{} model checking on paths} is also \PSPACE-complete for a fixed $\LLL^+(\Glob)$ \LTL{} formula. All these complexity lower bounds remain true even for DFAs over a constant-size alphabet.
\end{theorem}

%% file: sec-graph-traversal.tex
In this section, we prove \NP{}-completeness of \CS{\orep}, and explore the connections of the traversal relation $\orep$ with (classical) \LTL{} model checking. Let us describe the intuition of the proof. Observe that in \CS{\orep} only last visits of states matter, and for each $(p, q) \in R$ every path must visit $q$, even if it does not visit $p$. Hence, in some sense not visiting $p$ does not help with satisfying the restrictions. 
One can thus first find a synchronizing word for the DFA, forcing the paths starting from every its state to end in a particular state $f$, and then extend this word to satisfy all the restrictions on the order of last visits starting only from $f$. The latter is captured by a graph-theoretic problem \textsc{Last-Visits Traversal}. To show that this problem is NP-complete, and hence \CS{\orep} is NP-complete too, we look at the equivalent \textsc{First-Visits Traversal} problem, which is obtained by reversing all transitions and all restrictions of \textsc{Last-Visits Traversal}. This problem has a more intuitive nature, since now we are dealing with first visits of states, so it is more easier to control at every moment of time which states can already be visited, and which cannot yet be visited.

We start by formally defining \textsc{Last-Visits Traversal} and relating it to \CS{\orep}. This idea was already present in  \cite{Wolf2020}.

\begin{problem*}
		\textsc{Last-Visits Traversal} \\
		\textbf{Input:} A digraph $G = (V, E)$ and a relation $R \subseteq V \times V$. \\
		\textbf{Output:} Does there exist a path in $G$ such that for each $(v, u) \in R$
  \begin{itemize}
      \item this path visits $u$,
      \item and if it also visits $v$, then it must visit $u$ after the last time it visits $v$?
  \end{itemize}
\end{problem*}

\begin{lemma}\label{lemma:lvt-to-sync}
	The problem	\textsc{Last-Visits Traversal} for strongly connected digraphs can be reduced in polynomial time to \CS{\orep} for strongly connected synchronizing DFAs.
\end{lemma}
\begin{proof}
	Given an instance of \textsc{Last-Visits Traversal}, construct the following DFA $\Aa = (Q, \Sigma, \delta)$. Take $Q = V$, and take $\Sigma$ large enough to define $\delta$ in such a way that $G$ is the underlying digraph of~$\Aa$ (by possibly duplicating some edges of $G$ to ensure that the outdegrees of all the vertices are equal). Since $G$ is strongly connected, it has a subgraph $T$ on the whole set $V$ of its vertices such that $T$ is an oriented tree directed towards its root $r$. Add a fresh letter $s$ to $\Sigma$ and define its action according to the edges in $T$, and make it act induce a self-loop for $r$.
	
	Since $G$ is strongly connected, thus obtained automaton $\Aa$ is also strongly connected. It is also synchronizing, since the word $s^{|Q|}$ is synchronizing for it. We claim that there exists a path satisfying the requirements of \textsc{Last-Visits Traversal} if and only if there exists a synchronizing word $w$ for $\Aa$ such that $\orep(w)$ agrees with $R$.
	
	Indeed, if $w$ is a synchronizing word for $\Aa$ such that $\orep(w)$ agrees with $R$, then it is enough to take the path labelled by $w$ and starting in an arbitrary state of $\Aa$. Since for every transition from a state $q$ to a state $q'$ in $\Aa$ there is an edge $(q, q')$ in $G$, this path satisfies all the restrictions of \textsc{Last-Visits Traversal}.
	In the other direction, consider a path $\rho$ in $G$ satisfying all the restrictions of \textsc{Last-Visits Traversal}. For each edge $(q, q')$ of this path, there is a transition from $q$ to $q'$ in~$\Aa$. By concatenating all the letters labelling these transitions in the order of how they appear along the path, we obtain a word $w'$. Let $w$ be a synchronizing word for $\Aa$ mapping all its states the starting state of $\rho$. Since $\Aa$ is strongly connected, such a word exists. Then $ww'$ is a synchronizing word for $\Aa$ such that $\orep(ww')$ agrees with $R$. 
\end{proof}

To show that \CS{\orep} is in \NP{}, we will also need a reduction in the opposite direction.

\begin{lemma}\label{lemma:sync-to-lvt}
    The problem	\CS{\orep} can be reduced in polynomial time to \textsc{Last-Visits Traversal} for strongly connected digraphs.
\end{lemma}
\begin{proof}
    Let $\Aa = (Q, \Sigma, \delta)$ and $R \subseteq Q \times Q$ be the input of \CS{\orep}. First, we can assume that $\Aa$ is synchronizing, since it can be checked in polynomial time \cite{DBLP:conf/lata/Volkov08}. We can also assume that $R$ is non-empty. Observe that for every $(p, q) \in R$, the path starting in every state of $\Aa$ must visit $q$. Hence, $q$ must be reachable from every state of $\Aa$, so the underlying digraph of $\Aa$ can have only one maximal strongly connected component $C$ that is a sink, and this component must contain all the states occurring as second elements in the pairs in $R$. Let $w_1$ be a  synchronizing word for $\Aa$. This word necessarily sends all the states to a state in $C$, call this state $f$.
    
    We take the underlying digraph of $C$ and the restriction of $R$ to its set of vertices as the input of \textsc{Last-Visits Traversal}. Assume that there exists a required path $\rho$ for this problem, and let $g$ be the vertex where it begins. Let further $w_2$ be a word mapping $f$ to $g$, and $w_3$ be a word labeling the path $\rho$ in $\Aa$. Then the word $w_1w_2w_3$ is clearly a required word for \CS{\orep}. In the other direction, if there exists a word $w$ required in \CS{\orep}, then the path labeled by $w$ and starting in any vertex of $C$ satisfies the conditions of \textsc{Last-Visits Traversal}. \end{proof}

To prove that \textsc{Last-Visits Traversal} is \NP{}-complete, we consider a polynomially equivalent problem \textsc{First-Visits Traversal}. This problem is obtained by simply reversing all the edges of the digraph and changing all requirements for the required path accordingly. In other words, it is obtained by looking at the path in the reversed order, from the end to the beginning.

\begin{problem*}
\textsc{First-Visits Traversal} \\
\textbf{Input:} A digraph $G = (V, E)$ and a relation $R \subseteq V \times V$. \\
\textbf{Output:} Yes, if and only if there exist a path in $G$ such that for every $(u, v) \in R$
  \begin{itemize}
      \item this path visits $u$,
      \item and if it also visits $v$, then it must visit $u$ before visiting $v$.
  \end{itemize}
\end{problem*}

\begin{proposition}
	In the \textsc{First-Visits Traversal} problem, if there exists a required path, then there is such a path of length at most $(|V| - 1)^2$.  
\end{proposition}
\begin{proof}
Let $\rho$ be a shortest path with the required properties.	Let $S$ be the set of vertices $u$ such that for some vertex $v$ we have $(u, v) \in R$. Then $S$ is the set of vertices this path must visit. Assume that~$\rho$ visits for the first time the vertices in $S$ in the order $v_1, v_2, \ldots, v_{|S|}$. Represent $\rho$ as a concatenation of path $\rho_i$, $1 \le i \le |S| - 1$, so that $\rho_i$ is the segment of $\rho$ between the first visit of $v_i$ and the first visit of $v_{i + 1}$. We then can assume that for each $i$ the path $\rho_i$ does not have any cycles, since removing such cycles does not decrease the set of vertices $\rho$ can visit in the future. Hence, $\rho$ is a concatenation of at most $|V| - 1$ paths in $G$ containing no cycles, and so its length is at most $(|V| - 1)^2$.  
\end{proof}

\begin{corollary}\label{corr:fvt-in-np}
	\textsc{First-Visits Traversal} is in \NP{}.
\end{corollary}

If for each vertex $v$ there is at most one vertex $u$ with $(u, v) \in R$ (which will be the case in the construction in the proof of Proposition~\ref{prop-first-vis}), \textsc{First-Visits Traversal} has the following natural interpretation. By visiting a vertex $u$, we ``unlock'' all vertices $v$ such that $(u, v) \in R$, which are not allowed to be visited before $u$. Hence, the problem asks if, starting from some vertex, it is possible to visit a given set of vertices (defined as $\{u \mid \text{ there is } v \text { such that } (u, v) \in R\}$) in a graph where more and more vertices become available for visiting every time a new vertex in $S$ is visited. The proof of the following proposition uses this intuition.

\begin{proposition} \label{prop-first-vis}
The problem \textsc{First-Visits Traversal} is \NP{}-hard for strongly connected digraphs.
\end{proposition}

\begin{proof}
We reduce from the \NP-compete  \textsc{CNF-SAT} problem~\cite{Karp1972}.

\begin{problem*} \textsc{CNF-SAT} \\
\textbf{Input:} A set $C = \{c_1, \ldots, c_m\}$ of clauses over Boolean variables $x_1, \ldots, x_n$. \\
\textbf{Output:} Does there exist an assignment which satisfies all the clauses?
\end{problem*}

The intuition behind the reduction is as follows, see \Cref{ex-sat} and \Cref{fig-sat} for an illustration. With the traversal constraints, we force the path to start in the leftmost vertex. Then, in the left half of \Cref{fig-sat}, every choice of a vertex corresponds to assigning the value $0$ or $1$ to a variable. The right half of \Cref{fig-sat} corresponds to the clauses of the \textsc{CNF-SAT} formula. With the traversal constraints, we make sure that the vertex $c'_i$ corresponding to a clause $c_i$ can be reached if and only if the picked value of at least one of the variables in $c_{i-1}$ satisfies the clause $c_{i-1}$. In this case, the path can visit the vertex corresponding to this occurrence and proceed to the next clause. The same is true for $f_1$ and $c'_m$. After reaching the vertex $f_1$ in \Cref{fig-sat}, the path visits the vertices in the leftmost half that assign the values to the variables that are opposite to the values picked at the beginning.

Formally, given an instance of \textsc{CNF-SAT}, we construct the following graph $G = (V, E)$. For each variable~$x_i$, $1 \le i \le n$, we add three vertices $x'_i, x_i^0, x_i^1$ to $V$. For $1 \le i \le n - 1$, we add to $E$ the edges
$$(x'_i, x_i^0), (x'_i, x_i^1), (x_i^0, x'_{i + 1}), (x_i^1, x'_{i + 1}).$$

For each clause $c_j$, $1 \le j \le m$, we add a set of vertices $\{c'_j\} \cup \{c_j^k \mid x_k \in c_j \text{ or } \overline{x_k} \in c_j\}$ to $V$. For each $j$, $1 \le j \le m - 1$, we add to $E$ the edges
$$\{(c'_j, c_j^k) \mid x_k \in c_j  \text{ or } \overline{x_k} \in c_j\} \cup \{(c_j^k, c'_{j+1}) \mid x_k \in c_j  \text{ or } \overline{x_k} \in c_j\}.$$

We add to $V$ new vertices $f_1, f_2$. Finally, we add to $E$ the edges 
$$(x_n^0, c'_1), (x_n^1, c'_1), \{(c_m^k, f_1) \mid x_k \in c_j  \text{ or } \overline{x_k} \in c_j\}, (f_1, f_2), (f_2, x'_1),$$
and observe that thus constructed digraph $G$ is strongly connected. 

The relation $R$ is constructed as follows. For each clause $c_j$, we add to $R$ the pair $(x_i^1, c^i_j)$ if $x_i$ occurs in $c_j$, and the pair $(x_i^0, c^i_j)$ if $\overline{x_i}$ occurs in $c_j$. We also add the pairs $(x'_1, f_1)$ and $(f_1, f_2)$ to $R$. Example \ref{ex-sat} below illustrates this construction. We claim that there exists a path in $G$ satisfying the constraints $R$ if and only if the input of the \textsc{CNF-SAT} problem has a satisfying assignment.

In one direction, let $a_1, \ldots, a_n \in \{0, 1\}$ be an assignment satisfying $C$. Consider the path $\rho_1 = x'_1, x_1^{a_1}, x'_2, x_2^{a_2}, x'_3, \ldots, x'_n, x_n^{a_n}$. By construction of $R$, $c_j^k$ can be visited only after $x_k^h$ is visited, where $h = 0$ if $\overline{x_k} \in c_j$, and $h = 1$ if $x_k \in c_j$. Hence, since $a_1, \ldots, a_n$ satisfies $C$, any path starting with~$\rho_1$ has now, for each $1 \le j \le m$, at least one vertex $c_j^k$ that it can visit. Let $\rho_2$ be such a path starting in $c'_1$ and ending in $f_2$. It now remains to complete $\rho_1 \rho_2$ with a path $\rho_3$ starting in $x'_1$ that visits the vertices $x_i^{h_i}$, $1 \le i \le n$ and $h_i \in \{0, 1\}$, that were not visited by $\rho_1 \rho_2$. The path $\rho_1 \rho_2 \rho_3$ then satisfies the requirements of \textsc{First-Visits Traversal}.

In the other direction, let $\rho$ be a path in $G$ which satisfies the requirements of \textsc{First-Visits Traversal}. Since $(f_1, f_2) \in R$, state $f_1$ must be visited, but since $(x'_1, f_1) \in R$, $x'_1$ must be visited first. This means that $\rho$ must start in $x'_1$, since otherwise it is not possible to visit $x'_1$ before $f_1$. Let $\rho_1 \rho_2$ be the prefix of $\rho$ such that $\rho_1$ ends in $x_n^h$ for $h \in \{0, 1\}$ and does not visit $c'_1$, and $\rho_2$ starts in $c'_1$, ends in $f_1$ and does not visit $f_2$. For each $j$, $1 \le j \le m$, the path $\rho_2$ must visit one of the vertices $c_j^k$. By construction of $R$, this means that $\rho_1$ must first visit the vertex $x_k^1$ if  $x_k \in c_j$, and $x_k^0$ if $\overline{x_k} \in c_j$. For $1 \le i \le n$, take $a_i = 0$ if $\rho_1$ visits $x_i^0$, and $a_i = 1$ if $\rho_1$ visits $x_i^1$. We get that  $a_1, \ldots, a_n$ is a satisfying assignment for the input of \textsc{CNF-SAT}.
\end{proof}

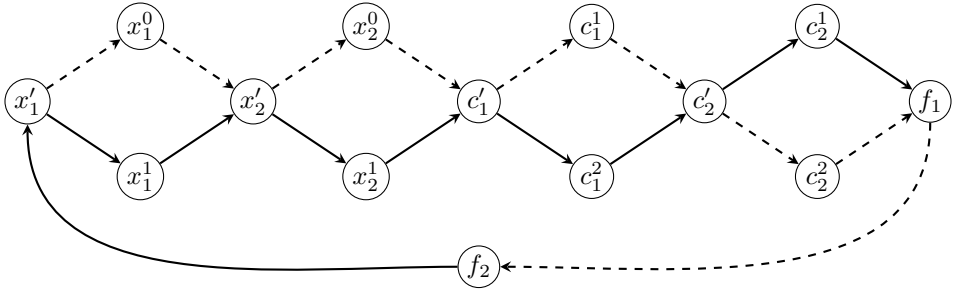
\begin{figure}[ht]\centering
\begin{tikzpicture} [node distance = 2cm]
\tikzset{every state/.style={inner sep=1pt,minimum size=1.5em}}

\node [state] at (1,-1) (x'1) {$x'_1$};
\node [state] at (2.5,0) (x10) {$x_1^0$};
\node [state] at (2.5,-2) (x11) {$x_1^1$};

\node [state] at (4,-1) (x'2) {$x'_2$};
\node [state] at (5.5,0) (x20) {$x_2^0$};
\node [state] at (5.5,-2) (x21) {$x_2^1$};

\node [state] at (7,-1) (c'1) {$c'_1$};
\node [state] at (8.5,0) (c11) {$c_1^1$};
\node [state] at (8.5,-2) (c12) {$c_1^2$};

\node [state] at (10,-1) (c'2) {$c'_2$};
\node [state] at (11.5,0) (c21) {$c_2^1$};
\node [state] at (11.5,-2) (c22) {$c_2^2$};

\node [state] at (13,-1) (f1) {$f_1$};
\node [state] at (7,-3.2) (f2) {$f_2$};

\path [-stealth, thick]
(x'1) edge [dashed] (x10)
(x'1) edge (x11)
(x10) edge [dashed] (x'2)
(x11) edge [] (x'2)

(x'2) edge [dashed] (x20)
(x'2) edge  (x21)
(x20) edge [dashed] (c'1)
(x21) edge [] (c'1)

(c'1) edge [dashed] (c11)
(c'1) edge (c12)

(c11) edge [dashed] (c'2)
(c12) edge (c'2)

(c'2) edge (c21)
(c'2) edge[dashed] (c22)

(c21) edge (f1)
(c22) edge[dashed] (f1)

(f1) edge [dashed, out=-90, in=0] (f2)

(f2) edge [out=-180, in=-90] (x'1)
;
\end{tikzpicture}
\caption{Illustration for Example \ref{ex-sat}.}\label{fig-sat}
\end{figure}

\begin{example}\label{ex-sat}
	Consider the CNF formula $(\overline{x}_1 \lor x_2) \land (x_1 \lor \overline{x}_2)$. The digraph from the reduction in the proof of Proposition \ref{prop-first-vis} is depicted in Figure \ref{fig-sat}. The set $R$ consists of the following pairs:
	$$(x_1^0, c_1^1), (x_2^1, c_1^2), (x_1^1, c_2^1), (x_2^0, c_2^2), (x'_1, f_1), (f_1, f_2).$$
	
	The path $\rho_1 \rho_2$ corresponding to the assignment $x_1 = 0, x_2 = 0$ is depicted by dashed edges. To satisfy the conditions, the extension of $\rho_1 \rho_2$ must visit $x_1^1$ and $x_2^1$.
\end{example}

Since \textsc{CNF-SAT} remains \NP{}-complete if every clause contains only three literals, we get that the digraph in the proof of \Cref{prop-first-vis} has at most three edges outgoing from each vertex. By combining \Cref{lemma:lvt-to-sync}, \Cref{lemma:sync-to-lvt}, \Cref{corr:fvt-in-np} and \Cref{prop-first-vis}, we thus obtain the main result of this section.

\begin{theorem}\label{thm:l-l-npc}
The problem \CS{\orep} is \NP{}-complete, even for strongly connected DFAs over a constant-size alphabet.
\end{theorem}

Finally, observe that in \textsc{First-Visits Traversal} the requirement on the path induced by $(u, v) \in R$ can be expressed in \LTL{} as $\neg v \Until u$. The \LTL{} model checking problem can be stated as follows. Given a digraph $G = (V, E)$ with a selected vertex $v$ and an \LTL{} formula $\phi$ over the set of atomic propositions $V$, does there exist a path in $G$ starting in $v$ that satisfies $\phi$? We get the following result from \Cref{corr:fvt-in-np} and \Cref{prop-first-vis}.

\begin{theorem}\label{thm:until-npc}
    \LTL{} model checking is \NP{}-complete for strongly connected digraphs and formulas which are conjunctions of formulas $\neg v \Until u$ with $u, v \in V$.
\end{theorem}

For formulas in $\LLL^+(\Until)$, LTL model checking is \PSPACE{}-complete \cite{Markey2004}, and hence it is \PSPACE{}-complete for such formulas in \LTL{} by the construction in the proof of Theorem 1 in \cite{DeGiacomo2013}. Hence, our much stronger constraints, including completely forbidding the nesting of until operators, drop the complexity for $\LLL^+(\Until)$, but only to \NP{}-complete.

%% file: sec-conclusions.tex
In this paper, we considered the problems of finding a word $w$ labeling paths satisfying a certain \LTL{} formula $\phi$ in a DFA $\Aa$. We concentrated on two cases: paths labeled by $w$ and starting in every state of $\Aa$, and the unique path in the power-set automaton of $\Aa$ labeled by $w$. We showed that both cases are solvable in polynomial space, and remain \PSPACE{}-complete even for very strong restrictions, e.g., for some fixed formulas.

One of the most natural questions is to find cases where these problems become tractable. This might be done by either considering a combination of restrictions on the DFA and the formula, or by considering different parameters and investigating the parameterized complexity of the problem. For example, observe that no \PSPACE{} lower bound in \Cref{sec:pspace-c} works for strongly connected DFAs. In combination with the requirement that the constraint relation has constant size, this might lead to polynomial time solvability. 

Another direction is to consider fragments of \LTL{} with limited modalities (such as $\LLL^+(\Glob)$) and limited nesting (such as the fragment considered in \Cref{thm:until-npc}). \Cref{subsec:pspace-path} provides such results for model checking on paths with formulas from $\LLL^+(\Event)$ and $\LLL^+(\Glob)$, but for model checking on sets the effect of such restrictions remains open.

Finally, there could be other natural ways to succinctly define DFAs. \Cref{sec:in-pspace} provides necessary conditions to perform model checking for them in \PSPACE{}, but it would be interesting to find some general conditions when the complexity of model checking over such succinctly defined DFAs drops below \PSPACE{}.